\documentclass[preprint,12pt]{elsarticle}

\usepackage{amssymb}
\usepackage{amsmath}
\usepackage{amsthm}
\usepackage{booktabs}
\usepackage{bm}
\usepackage{color}

\newtheorem{theorem}{Theorem}
\newtheorem{corollary}{Corollary}

\journal{Physica D: Nonlinear Phenomena}

\begin{document}

\begin{frontmatter}

%% ------------------------------------------------------------------
%% TODO (author to fill in): title, author names, affiliation(s).
%% ------------------------------------------------------------------
\title{Jordan-Block Degeneracy and Cubic-Order Bifurcating
Periodic Orbits in Minimum-Energy Optimal Control of Hamiltonian
Equilibria}

\author[1]{Mai Bando\corref{cor1}}
\author[1,2]{Ayano Tsuruta}
\author[1]{Shanshan Pan}
\author[3]{Daniel J. Scheeres}

\cortext[cor1]{Corresponding author.}

\affiliation[1]{
    organization={Department of Aeronautics and Astronautics, Kyushu University},
    addressline={744 Motooka, Nishi-ku},
    city={Fukuoka},
    postcode={819-0395},
    state={Fukuoka},
    country={Japan}
}

\affiliation[2]{
    organization={Human Spaceflight Technology Directorate,
                  Japan Aerospace Exploration Agency (JAXA)},
    addressline={2-1-1 Sengen},
    city={Tsukuba},
    postcode={305-8505},
    state={Ibaraki},
    country={Japan}
}

\affiliation[3]{
    organization={Ann and H.J. Smead Department of Aerospace Engineering Sciences,
                  University of Colorado Boulder},
    city={Boulder},
    postcode={80309},
    state={Colorado},
    country={USA}
}

\begin{abstract}
Equilibria of the Hamiltonian system associated with Pontryagin's minimum principle exhibit an exact doubling of the natural spectrum and, under a simple pairing condition, a Jordan block at every simple purely imaginary eigenvalue.
Consequently, the classical Lyapunov Center Theorem does not apply to the augmented system, and no periodic orbit with nonzero optimal control bifurcates at linear order.
We establish this mechanism in general and show that an optimal-control-induced periodic family emerges at cubic order in a Lindstedt--Poincar\'e expansion.
The mechanism is illustrated in closed form for the pendulum and evaluated numerically for the planar $L_2$ equilibrium of Hill's restricted three-body problem, where the third-order approximation is validated against an independently computed family of periodic orbits.
\end{abstract}

\begin{highlights}
\item The Pontryagin augmentation associated with minimum-energy optimal control exactly doubles the natural spectrum of a Hamiltonian system at an equilibrium.

\item An explicit pairing condition determines whether a doubled eigenvalue forms a Jordan block.

\item The resulting Jordan block prevents the direct application of the classical Lyapunov Center Theorem.

\item An optimal-control-induced family of periodic solutions
bifurcates at cubic order, as demonstrated for the planar pendulum and
the Hill3BP $L_2$ equilibrium.
\end{highlights}


\end{frontmatter}

\section{Introduction}

Quadratic optimal control problems play a central role in modern control theory and trajectory optimization. Through Pontryagin's minimum principle, such problems are transformed into Hamiltonian boundary-value systems by introducing adjoint variables in addition to the physical state \cite{pontryagin1962,bryson1975,liberzon2012,agrachev2004}. 
This Hamiltonian formulation provides the theoretical foundation for
indirect optimal control methods \cite{bryson1975} and has been widely applied to
trajectory optimization in astrodynamics and aerospace engineering \cite{Conway2010,Kelly2017}.

For conservative mechanical systems, the uncontrolled dynamics is itself Hamiltonian. Applying Pontryagin's minimum principle therefore produces an augmented Hamiltonian system defined on the product of the state and adjoint spaces.
The augmented Hamiltonian system has also been investigated from a
dynamical viewpoint, particularly in connection with invariant manifolds
\cite{Sakamoto2008} and turnpike phenomena \cite{trelat2018steady}.
However, these studies mainly concern hyperbolic equilibria or periodic
solutions. Recent work has also examined center dynamics in augmented Hamiltonian
systems \cite{sakamoto2024}, but the non-semisimple center dynamics
induced by Pontryagin augmentation and the resulting bifurcation of
periodic extremals remain largely unexplored.

The local dynamics near Hamiltonian equilibria has been extensively studied in classical dynamical systems theory. In particular, the Lyapunov Center Theorem guarantees the existence of one-parameter families of periodic orbits bifurcating from a non-resonant equilibrium possessing a simple purely imaginary eigenvalue pair \cite{lyapunov1992,weinstein1973,moser1976,meyer2009}. 
This theorem is one of the fundamental results in the local theory of Hamiltonian systems and forms the basis for the local construction of Lyapunov and halo orbit families in celestial mechanics \cite{richardson1980,gomez2001}. 
Its applicability, however, relies on the semisimplicity of the center eigenvalue. Once a Jordan block is present, the classical Lyapunov-center bifurcation no longer applies, and the existence of periodic solutions must instead be determined through higher-order nonlinear analysis.

Non-semisimple eigenvalue configurations at Hamiltonian equilibria have
been studied extensively in the bifurcation literature. A classical
example is the Hamiltonian Hopf bifurcation, in which two pairs of
purely imaginary eigenvalues collide as a system parameter is varied.
At the critical parameter value, the repeated eigenvalues may become
non-semisimple and form Jordan blocks, in which case the local dynamics
is described by the non-semisimple normal-form theory developed by
van der Meer \cite{vanderMeer1985}.
Related normal-form studies have also considered non-semisimple purely imaginary resonances, including the non-semisimple 1:1 resonance \cite{vanGils1990}.

In these classical bifurcation settings, non-semisimplicity arises at
critical parameter values as system parameters are varied. In contrast,
the degeneracy considered in this study does not arise through parameter
variation; instead, it is generated structurally by the Pontryagin
augmentation associated with minimum-energy optimal control.
We show that the Hamiltonian system associated with Pontryagin's minimum
principle possesses the following structural property: every eigenvalue
of the linearization at an equilibrium of the underlying Hamiltonian
system appears with doubled algebraic multiplicity in the augmented
system. Furthermore, under a simple pairing condition, the doubled
purely imaginary eigenvalues considered here form Jordan blocks rather
than remaining semisimple. Consequently, the assumptions of the
classical Lyapunov Center Theorem are violated, and no nontrivial
periodic solution with a nonzero adjoint component bifurcates at linear
order.

Although the linearized system admits no periodic family with a nonzero
adjoint component, such a family may emerge through nonlinear effects.
We show that an optimal-control-induced periodic family is selected by a
bifurcation condition that first appears at cubic order in a
Lindstedt--Poincar\'e expansion.
The cubic-order bifurcation mechanism is first derived analytically for
the planar pendulum, for which all coefficients are obtained in closed
form, and is then examined numerically for the planar $L_2$
equilibrium of Hill's restricted three-body problem.
The resulting third-order approximation is validated against an
independently computed family of periodic orbits.

The remainder of the paper is organized as follows.
Section~\ref{sec:theory} establishes the general spectral structure of
the Pontryagin-augmented Hamiltonian system and derives the pairing
condition that determines the formation of Jordan blocks.
Section~\ref{sec:pendulum} presents the complete cubic-order analysis
for the pendulum.
Section~\ref{sec:hill3bp} applies the theory to the planar $L_2$
equilibrium of Hill's restricted three-body problem and compares the
analytical predictions with numerical periodic-orbit families.
Finally, Section~\ref{sec:discussion} discusses the implications of the
present results, and Section~\ref{sec:conclusion} concludes the paper.
\section{General theory}
\label{sec:theory}

\subsection{Problem setting}

Consider the control-affine system
\begin{equation}
    \dot{\bm{x}}
    =
    f(\bm{x})+\bm{B}\bm{u},
\end{equation}
where $\bm{x}\in\mathbb{R}^{m}$ denotes the state vector,
$\bm{u}\in\mathbb{R}^{r}$ the control input,
$f:\mathbb{R}^{m}\to\mathbb{R}^{m}$ a sufficiently smooth vector field,
and $\bm{B}\in\mathbb{R}^{m\times r}$ a constant input matrix.
Assume that the uncontrolled system possesses an equilibrium
$\bm{x}^*\in\mathbb{R}^{m}$ satisfying
\begin{equation}
f(\bm{x}^*)=\bm{0}.
\end{equation}
We consider the fixed-terminal-time minimum-energy optimal control problem with cost functional
\begin{equation}
J=\frac12\int_0^{t_f}|\bm{u}|^2\,dt.
\end{equation}
subject to
\begin{equation}
\dot{\bm{x}}=f(\bm{x})+\bm{B}\bm{u}.
\end{equation}
The subsequent analysis depends only on the associated Hamiltonian system and is independent of the boundary conditions of the optimal control problem.
Applying Pontryagin's minimum principle\cite{pontryagin1962,bryson1975}, the associated Hamiltonian is
\begin{equation}
H(\bm{x},\bm{p},\bm{u})=\bm{p}^{\mathrm T}
\left(f(\bm{x})+\bm{B}\bm{u}\right)+\frac12\bm{u}^{\mathrm T}\bm{u},
\end{equation}
where $\bm{p}\in\mathbb{R}^{m}$ denotes the adjoint variable.
The stationarity condition yields
\begin{equation}
\bm{u}^*= -\bm{B}^{\mathrm T}\bm{p}.
\end{equation}
Substituting this expression into the canonical equations gives the state--adjoint system
\begin{align}
\dot{\bm{x}}&=f(\bm{x})-\bm{B}\bm{B}^{\mathrm T}\bm{p},\\
\dot{\bm{p}}&=-Df(\bm{x})^{\mathrm T}\bm{p}.
\end{align}
Since $\bm{x}^*$ is an equilibrium of the uncontrolled system,
$f(\bm{x}^*)=\bm{0}$, the corresponding equilibrium of the
Pontryagin system is $(\bm{x}^*,\bm{0})$.
Linearizing about $(\bm{x}^*,\bm{0})$, with
$\bm A=Df(\bm{x}^*)$, gives
\begin{equation}
\begin{pmatrix}
\dot{\delta\bm{x}}\\
\dot{\delta\boldsymbol p}
\end{pmatrix}
=
\begin{pmatrix}
\boldsymbol A & -\boldsymbol B\boldsymbol B^{\mathrm T}\\
\boldsymbol0 & -\boldsymbol A^{\mathrm T}
\end{pmatrix}
\begin{pmatrix}
\delta\bm{x}\\
\delta\boldsymbol p
\end{pmatrix} =: \boldsymbol M\begin{pmatrix}
\delta\bm{x}\\
\delta\boldsymbol p
\end{pmatrix}.
\label{eq:M}
\end{equation}
The matrix $\bm M$ is block upper triangular for any smooth vector field
$f$ and any constant input matrix $\bm B$ under the minimum-energy
control formulation considered here. This structure follows from
linearization of the state--adjoint equations about the equilibrium
$(\bm{x}^*,\bm0)$ and forms the starting point of the subsequent
analysis. The additional assumption that the uncontrolled dynamics is
Hamiltonian is imposed in the next subsection, where the spectral
consequences of Eq.~\eqref{eq:M} are established.

\subsection{Structural spectral doubling}
\begin{theorem}[Structural spectral doubling]
Let
\[
\boldsymbol M=
\begin{pmatrix}
\boldsymbol A & -\boldsymbol B\boldsymbol B^{\mathrm T}\\
\boldsymbol0 & -\boldsymbol A^{\mathrm T}
\end{pmatrix},
\]
where $\boldsymbol A\in\mathbb R^{m\times m}$ and
$\boldsymbol B\in\mathbb R^{m\times r}$.
Then
\begin{equation}
\mathrm{Spec}(\boldsymbol M)
=
\mathrm{Spec}(\boldsymbol A)
\cup
\bigl(-\mathrm{Spec}(\boldsymbol A)\bigr).
\label{eq:doubling}
\end{equation}
If, in addition, the uncontrolled dynamics is Hamiltonian, then
$\boldsymbol A$ is a Hamiltonian matrix whose spectrum is symmetric
with respect to the origin. Consequently, $\boldsymbol M$ has the same
distinct eigenvalues as $\boldsymbol A$, with each eigenvalue having
exactly twice its algebraic multiplicity in $\boldsymbol M$.
\end{theorem}

\begin{proof}
Since $\boldsymbol M$ is block upper triangular, its spectrum is the union of the
spectra of the diagonal blocks,
\[
\mathrm{Spec}(\boldsymbol M) = \mathrm{Spec}(\boldsymbol A)
\cup \mathrm{Spec}(-\boldsymbol A^{\mathrm T}).
\]
Because transposition preserves eigenvalues, $\mathrm{Spec}(-\boldsymbol A^{\mathrm T}) = -\mathrm{Spec}(\boldsymbol A)$, which proves Eq.~\eqref{eq:doubling}.

If $\boldsymbol A$ is Hamiltonian, then $\boldsymbol A=\boldsymbol J\boldsymbol S$, where $\boldsymbol J$ is the canonical symplectic matrix and $\boldsymbol S$ is symmetric.
Using $\boldsymbol J^{-1}\boldsymbol A^{\mathrm T}\boldsymbol J = -\boldsymbol A$,
it follows that $\boldsymbol A^{\mathrm T}$ and $-\boldsymbol A$ are similar.
Since transposition preserves the spectrum, $\boldsymbol A$ and
$-\boldsymbol A$ therefore possess the same spectrum.
Hence
\[
\mathrm{Spec}(\boldsymbol M) = \mathrm{Spec}(\boldsymbol A) \cup \mathrm{Spec}(\boldsymbol A),
\]
which means that every eigenvalue of $\boldsymbol A$ appears in $\boldsymbol M$
with doubled algebraic multiplicity.
\end{proof}

\begin{corollary}[Absence of spectral doubling in dissipative systems]
\label{cor:dissipative}

Suppose that $\boldsymbol{x}^*$ is a hyperbolic asymptotically stable
equilibrium of the uncontrolled system, so that
\[
\operatorname{Re}(\lambda)<0, \qquad \forall\lambda\in\operatorname{Spec}(\boldsymbol A).
\]
Then
\[
\operatorname{Spec}(\boldsymbol A) \cap \bigl(-\operatorname{Spec}(\boldsymbol A)\bigr) = \emptyset.
\]
Therefore, no eigenvalue of $\boldsymbol A$ coincides with an eigenvalue
of $-\boldsymbol A^{\mathrm T}$, and the spectral doubling and associated
Jordan-block degeneracy considered in this paper do not occur.

\end{corollary}

\subsection{Jordan-block criterion}\label{subsec:criterion}
The previous theorem shows that every eigenvalue of the natural Hamiltonian equilibrium is doubled in the augmented system. The remaining question is whether the doubled eigenvalue is semisimple or forms a Jordan block. 

\begin{theorem}[Jordan-block criterion]
Suppose that the uncontrolled dynamics is Hamiltonian, and let
$\lambda=i\omega$ be a simple eigenvalue of $\boldsymbol A$.
Let $\boldsymbol\ell$ and $\boldsymbol p_0$ be left eigenvectors of
$\boldsymbol A$ associated with $\lambda$ and $-\lambda$,
respectively, satisfying
\[
\boldsymbol A^{\mathrm T}\boldsymbol\ell
=
\lambda\boldsymbol\ell,
\qquad
\boldsymbol A^{\mathrm T}\boldsymbol p_0
=
-\lambda\boldsymbol p_0.
\]
Then the doubled eigenvalue $\lambda$ of $\boldsymbol M$ satisfies
\begin{align}
\lambda \text{ forms a }2\times2\text{ Jordan block}
&\Longleftrightarrow
\boldsymbol\ell^{\mathrm T}
\boldsymbol B\boldsymbol B^{\mathrm T}
\boldsymbol p_0
\neq0,
\\
\lambda \text{ is semisimple}
&\Longleftrightarrow
\boldsymbol\ell^{\mathrm T}
\boldsymbol B\boldsymbol B^{\mathrm T}
\boldsymbol p_0
=0.
\end{align}
\end{theorem}
\begin{proof}
An eigenvector of $\boldsymbol M$ associated with $\lambda$ is always given by
$(\boldsymbol v,\boldsymbol 0)$, where
\[
\boldsymbol A\boldsymbol v=\lambda\boldsymbol v.
\]
Suppose that a second linearly independent eigenvector
$(\boldsymbol x,\boldsymbol p)$ exists.
From the lower block of Eq.~\eqref{eq:M},
\[
-\boldsymbol A^{\mathrm T}\boldsymbol p
=
\lambda\boldsymbol p,
\]
or equivalently,
\[
\boldsymbol A^{\mathrm T}\boldsymbol p
=
-\lambda\boldsymbol p.
\]
Since $-\lambda$ is a simple eigenvalue of $\boldsymbol A^{\mathrm T}$,
$\boldsymbol p$ must be a scalar multiple of
$\boldsymbol p_0$.
Without loss of generality, we therefore set
$\boldsymbol p=\boldsymbol p_0$.
The upper block then becomes
\[
(\boldsymbol A-\lambda\boldsymbol I)\boldsymbol x
=
\boldsymbol B\boldsymbol B^{\mathrm T}\boldsymbol p_0.
\]
Since $\lambda$ is a simple eigenvalue of $\boldsymbol A$,
the left nullspace of
$\boldsymbol A-\lambda\boldsymbol I$
is one-dimensional and is spanned by
$\boldsymbol\ell$. Thus,
\[
\operatorname{Null}\!\left(
(\boldsymbol A-\lambda\boldsymbol I)^{\mathrm T}
\right)
=
\operatorname{span}\{\boldsymbol\ell\},
\]
with
\[
\boldsymbol\ell^{\mathrm T}
(\boldsymbol A-\lambda\boldsymbol I)
=
0.
\]
Using the fundamental orthogonality relation
\[
\operatorname{Range}(\boldsymbol A-\lambda\boldsymbol I)
=
\operatorname{Null}\!\left(
(\boldsymbol A-\lambda\boldsymbol I)^{\mathrm T}
\right)^\perp
=
\operatorname{span}\{\boldsymbol\ell\}^{\perp},
\]
the equation
\[
(\boldsymbol A-\lambda\boldsymbol I)\boldsymbol x
=
\boldsymbol B\boldsymbol B^{\mathrm T}\boldsymbol p_0
\]
is solvable if and only if
\[
\boldsymbol B\boldsymbol B^{\mathrm T}\boldsymbol p_0
\in
\operatorname{span}\{\boldsymbol\ell\}^{\perp}.
\]
Equivalently,
\[
\boldsymbol\ell^{\mathrm T}
\boldsymbol B\boldsymbol B^{\mathrm T}
\boldsymbol p_0
=
0.
\]
Hence, a second linearly independent eigenvector associated with
$\lambda$ exists if and only if the pairing
\[
\boldsymbol\ell^{\mathrm T}
\boldsymbol B\boldsymbol B^{\mathrm T}
\boldsymbol p_0
\]
vanishes. Therefore, if the pairing vanishes, the doubled eigenvalue
$\lambda$ is semisimple. If the pairing is nonzero, the geometric
multiplicity of $\lambda$ is one, whereas its algebraic multiplicity is two.
Consequently, $\lambda$ forms a $2\times 2$ Jordan block.
\end{proof}

\begin{corollary}[Inapplicability of the Lyapunov Center Theorem]
Let $\boldsymbol\xi^*$ be an equilibrium of a Hamiltonian system at
which the classical Lyapunov Center Theorem applies, namely,
$\boldsymbol A$ possesses a simple non-resonant purely imaginary
eigenvalue pair $\pm i\omega$.
If
\[
\boldsymbol\ell^{\mathrm T}
\boldsymbol B\boldsymbol B^{\mathrm T}
\boldsymbol p_0
\neq0,
\]
then the doubled eigenvalue $\lambda=i\omega$ of the augmented system is
non-semisimple and forms a $2\times2$ Jordan block. Consequently, the
classical Lyapunov Center Theorem cannot be directly applied to this
eigenvalue pair. Moreover, the linearized augmented system admits no
periodic solution associated with this eigenvalue pair that has a
nonzero adjoint component.
\end{corollary}

This situation is examined in the following sections.
For both the pendulum and the $L_2$ equilibrium of
Hill's restricted three-body problem, the cubic-order solvability
conditions yield a nontrivial control-induced periodic family,
confirming the nonlinear bifurcation predicted by the above
corollary.

\section{Application I: the pendulum}
\label{sec:pendulum}

This section applies the general theory developed in
Section~\ref{sec:theory} to the pendulum.
The complete reduction and the resulting bifurcation analysis can be
carried out in closed form.

\subsection{Optimal control formulation}
\label{sec:pendulum-setup}

Consider the pendulum
\begin{equation}
\dot{\bm{x}}
=
f(\bm{x})
+
\bm{B}u,
\qquad
\bm{x}
=
\begin{pmatrix}
q_1\\
q_2
\end{pmatrix},
\qquad
\bm{B}
=
\begin{pmatrix}
0\\
1
\end{pmatrix},
\end{equation}
where
\begin{equation}
f(\bm{x})
=
\begin{pmatrix}
q_2\\
-k^2\sin q_1
\end{pmatrix},
\qquad
\left(k=\sqrt{\frac{g}{l}}\right).
\end{equation}
The minimum-energy optimal control problem is
\begin{equation}
J
=
\frac12
\int_{t_0}^{t_f}
u^2\,dt.
\label{eq:J}
\end{equation}
Introducing the adjoint variable
\(
\bm{p}=(p_1,p_2)^{\mathrm T}
\)
through Pontryagin's minimum principle gives the Hamiltonian
\begin{equation}
H(\bm{x},\bm{p},u)
=
p_1q_2
+
p_2(-k^2\sin q_1+u)
+
\frac12u^2.
\end{equation}
The stationarity condition 
% \[
% \frac{\partial H}{\partial u}=0
% \]
yields
\begin{equation}
u^*=-p_2.
\end{equation}
Substituting the optimal control gives the state--adjoint system
\begin{align}
\dot{\bm{x}}
&=
\left(
\frac{\partial H}{\partial\bm{p}}
\right)^{\mathrm T}
=
\begin{pmatrix}
q_2\\
-k^2\sin q_1-p_2
\end{pmatrix},
\label{eq:canonical_q}
\\
\dot{\bm{p}}
&=
-
\left(
\frac{\partial H}{\partial\bm{x}}
\right)^{\mathrm T}
=
\begin{pmatrix}
k^2 p_2\cos q_1\\
-p_1
\end{pmatrix}.
\label{eq:canonical_p}
\end{align}
The equilibrium corresponding to the downward pendulum is
$(\bm{x},\bm{p})=(\bm{0},\bm{0})$.
Linearizing Eqs.~\eqref{eq:canonical_q}--\eqref{eq:canonical_p} about
this equilibrium gives
\begin{align}
\delta\dot{\bm z}
&=
\bm{M}\,\delta\bm z,
\qquad
\delta\bm z
=
\begin{pmatrix}
\delta\bm{x}\\
\delta\bm{p}
\end{pmatrix},
\\
\bm{M}
&=
\begin{pmatrix}
0 & 1 & 0 & 0\\
-k^2 & 0 & 0 & -1\\
0 & 0 & 0 & k^2\\
0 & 0 & -1 & 0
\end{pmatrix}.
\label{eq:pendulum_linear}
\end{align}
In the following, the parameter is normalized to $k=1$.

\paragraph{Verification of the Jordan-block criterion}
For $k=1$, the linearization of the uncontrolled pendulum dynamics is
\begin{equation}
\bm A
=
\begin{pmatrix}
0 & 1\\
-1 & 0
\end{pmatrix},
\qquad
\bm B=
\begin{pmatrix}
0\\
1
\end{pmatrix},
\end{equation}
with eigenvalues $\lambda=\pm i$.
The corresponding augmented matrix
$\bm M$ is obtained from Eq.~\eqref{eq:pendulum_linear}
by setting $k=1$. Its characteristic polynomial is
\[
(\lambda^2+1)^2=0,
\]
so that the eigenvalues $\lambda=\pm i$ each have algebraic
multiplicity two.
For $\lambda=i$, left eigenvectors of $\bm A$ associated
with $\lambda$ and $-\lambda$ may be chosen as
\[
\bm\ell=
\begin{pmatrix}
i\\
1
\end{pmatrix},
\qquad
\bm p_0=
\begin{pmatrix}
-i\\
1
\end{pmatrix}.
\]
Hence,
\[
\bm\ell^{\mathrm T}\bm B\bm B^{\mathrm T}\bm p_0
=1\neq0.
\]
Therefore, the Jordan-block criterion implies that the doubled eigenvalues are non-semisimple.
Indeed, the Jordan normal form is
\begin{equation}
\bm{J}
=
\begin{pmatrix}
-i & 1 & 0 & 0\\
0 & -i & 0 & 0\\
0 & 0 & i & 1\\
0 & 0 & 0 & i
\end{pmatrix},
\end{equation}
showing that each eigenvalue has geometric multiplicity one.
The associated generalized eigenvectors produce secular terms $te^{\pm it}$, corresponding to oscillations with linearly growing amplitudes.

\subsection{Exact reduction to a second-order system}
The state--adjoint system is exactly equivalent to
\begin{equation}
\ddot q_1+\sin q_1=-p_2, \qquad \ddot p_2+p_2\cos q_1=0,
\label{eq:reduced_pair}
\end{equation}
with
\[
q_2=\dot q_1, \qquad p_1=-\dot p_2.
\]
Equation~\eqref{eq:reduced_pair} reveals the resonance mechanism directly.
Linearization of Eq.~\eqref{eq:reduced_pair} shows that both oscillators
have the same natural frequency $\omega=1$. The physical coordinate
$q_1$ is directly forced by the adjoint variable $p_2$, whereas the
effect of $q_1$ on the $p_2$ dynamics enters only through the nonlinear
coefficient $\cos q_1$ and therefore does not appear at linear order.
Setting $p_2\equiv0$ removes the adjoint dynamics and reduces Eq.~\eqref{eq:reduced_pair} to the classical nonlinear pendulum,
\[
\ddot q_1+\sin q_1=0,
\]
which possesses the familiar one-parameter family of periodic solutions constituting the trivial branch with zero optimal control. 
The following subsection analyzes the nonlinear bifurcation from this
trivial branch using a Lindstedt--Poincar\'e expansion.

\subsection{Nonlinear analysis}
\label{sec:pendulum-nonlinear}
The analysis follows the approach of Richardson~\cite{richardson1980}, who constructed small-amplitude periodic orbits around the collinear libration points of the circular restricted three-body problem by introducing a strained time variable and expanding both the solution and its oscillation frequency.

\subsubsection{Lindstedt--Poincar\'e expansion}
Under the amplitude-reversal symmetry $\varepsilon\rightarrow-\varepsilon$, the solution components $q_1$ and $p_2$ are expanded in odd powers of $\varepsilon$, whereas the frequency $\omega$ is expanded in even powers.
Furthermore, the linear analysis shows that a periodic solution cannot contain a nonzero $p_2$ component at $O(\varepsilon)$.
Since all even powers are excluded by symmetry, the first possible
nonzero adjoint contribution therefore occurs at cubic order.
Accordingly, we assume
\begin{align}
q_1&=\varepsilon Q_1(\tau)+\varepsilon^3Q_3(\tau)+O(\varepsilon^5),\label{eq:q1}\\
p_2&=\varepsilon^3P_3(\tau)+\varepsilon^5P_5(\tau)+O(\varepsilon^7),\label{eq:p2}\\
\omega&=1+\varepsilon^2\omega_2+O(\varepsilon^4),
\label{eq:omega}
\end{align}
where $\tau=\omega t$ is the rescaled time variable accounting for the amplitude-dependent frequency. 
We choose the initial displacement $\varepsilon=q_1(0)$ as the
amplitude parameter and set the phase origin at a symmetry point of
the periodic orbit. Accordingly, we impose the normalization and
symmetry conditions
\begin{equation}
Q_1(0)=1,\qquad Q_3(0)=0,\qquad
Q_1'(0)=0,\qquad P_3'(0)=0.
\label{eq:symmetry_conditions}
\end{equation}
%---------------------------------------------------%
\paragraph{First order ($O(\varepsilon)$)}
%---------------------------------------------------%
Substituting $\tau=\omega t$ into Eqs.~\eqref{eq:reduced_pair} gives
\begin{align}
&\omega^2 q_1''+\sin q_1 = -p_2,\label{eq:reduced_tau1}\\
&\omega^2 p_2''+p_2\cos q_1 = 0,
\label{eq:reduced_tau2}
\end{align}
where the prime denotes differentiation with respect to $\tau$.
Substituting Eqs.~\eqref{eq:q1}--\eqref{eq:omega} gives
\begin{align}
\omega^2 &= 1+2\varepsilon^2\omega_2+O(\varepsilon^4),\\
q_1'' &= \varepsilon Q_1'' +\varepsilon^3Q_3'' +O(\varepsilon^5),\\
\sin q_1 &= \varepsilon Q_1 +\varepsilon^3 \left( Q_3-\frac16Q_1^3 \right) +O(\varepsilon^5).
\end{align}
Since $p_2=O(\varepsilon^3)$, the adjoint forcing does not contribute
at $O(\varepsilon)$. The leading-order equation is therefore
\begin{equation}
Q_1''+Q_1=0.
\end{equation}
The normalization and symmetry conditions \eqref{eq:symmetry_conditions} then give the leading-order solution
\begin{equation}
Q_1(\tau)=\cos\tau.
\end{equation}

%---------------------------------------------------%
\paragraph{Third order ($O(\varepsilon^3)$)}
%---------------------------------------------------%
Collecting the $O(\varepsilon^3)$ terms gives
\begin{equation}
Q_3''+Q_3
=
2\omega_2\cos\tau
+\frac16\cos^3\tau
-P_3(\tau).\label{eq:Q3}
\end{equation}
At the same order, the $p_2$ equation is homogeneous:
\begin{equation}
P_3''+P_3=0.
\end{equation}
Using the symmetry condition $P_3'(0)=0$ in Eq.~\eqref{eq:symmetry_conditions} gives
\begin{equation}
P_3(\tau)=B\cos\tau,
\label{eq:P_3}
\end{equation}
where $B$ is an undetermined adjoint amplitude.
Using $\cos^3\tau =
\frac{3\cos\tau+\cos3\tau}{4}$, Eq.~\eqref{eq:Q3} becomes
\begin{equation}
Q_3''+Q_3 = \left( 2\omega_2-B+\frac18 \right)\cos\tau +\frac1{24}\cos3\tau.
\end{equation}
The forcing proportional to $\cos\tau$ is resonant with the homogeneous solution and would generate a secular term proportional to $\tau\sin\tau$.
Periodic solutions therefore require
\begin{equation}
2\omega_2-B+\frac18=0.
\label{eq:sc1}
\end{equation}
This provides the first relation between the adjoint amplitude $B$
and the frequency correction $\omega_2$.

\paragraph{Fifth order ($O(\varepsilon^5)$)}
Collecting the $O(\varepsilon^5)$ terms gives
\begin{equation}
P_5''+P_5 = -2\omega_2P_3'' +\frac12P_3Q_1^2.
\end{equation}
Using $P_3=B\cos\tau$, $P_3''=-B\cos\tau$ and $Q_1=\cos\tau$, the equation becomes
\begin{equation}
P_5''+P_5 = B\left(
2\omega_2+\frac38 \right)\cos\tau +\frac{B}{8}\cos3\tau.
\end{equation}
As at third order, the resonant fundamental harmonic must vanish to avoid a secular term proportional to $\tau\sin\tau$. Hence,
\begin{equation}
B\left(
2\omega_2+\frac38
\right)=0.
\label{eq:sc2}
\end{equation}
The solution $B=0$ gives $P_3\equiv0$ and corresponds to the continuation of the uncontrolled pendulum family.
For the nontrivial branch ($B\neq0$), Eqs.~\eqref{eq:sc1} and \eqref{eq:sc2} give
\begin{equation}
\omega_2=-\frac3{16}, \qquad B=-\frac14.
\end{equation}
Therefore,
\begin{equation}
P_3(\tau)=-\frac14\cos\tau,
\end{equation}
and the leading-order adjoint amplitude is
\begin{equation}
p_2(0) = -\frac14\varepsilon^3 +O(\varepsilon^5).
\label{eq:pendulum_cubic_law}
\end{equation}
The corresponding frequency and period are
\begin{align}
\omega(\varepsilon)
&=
1-\frac3{16}\varepsilon^2+O(\varepsilon^4),
\\
T(\varepsilon)
&=
\frac{2\pi}{\omega(\varepsilon)}
=
2\pi
\left(
1+\frac3{16}\varepsilon^2
\right)
+O(\varepsilon^4).
\end{align}
Since $\sin q_1=q_1-\frac16q_1^3+O(q_1^5)$, the first nonlinear
correction to the pendulum restoring force appears at
$O(\varepsilon^3)$.
Accordingly, the nontrivial periodic branch bifurcates from the
zero-adjoint branch with the cubic scaling
$p_2(0)=-\varepsilon^3/4+O(\varepsilon^5)$.
The cubic-order coefficient is determined by the two
solvability conditions \eqref{eq:sc1} and \eqref{eq:sc2}.

\subsubsection{Validation of the analytical approximation}
\label{sec:pendulum-numerical}

To assess the accuracy of the cubic-order analytical approximation, the analytical predictions are compared with numerically computed periodic orbits. Figure~\ref{fig:po_family_origin_pend} shows the periodic-orbit family around the origin.
Figure~\ref{fig:p2_compare_pend} compares the analytical and numerical initial costates as a function of the initial amplitude $\varepsilon=q_1(0)$, while Fig.~\ref{fig:period_compare_pend} compares the corresponding periods.
In both cases, the analytical predictions agree well with the numerical results for small amplitudes. As the amplitude increases, neglected higher-order nonlinear terms gradually become significant, leading to slight deviations from the numerical solutions.

\begin{figure}[tb]
\centering
\includegraphics[width=0.8\linewidth]{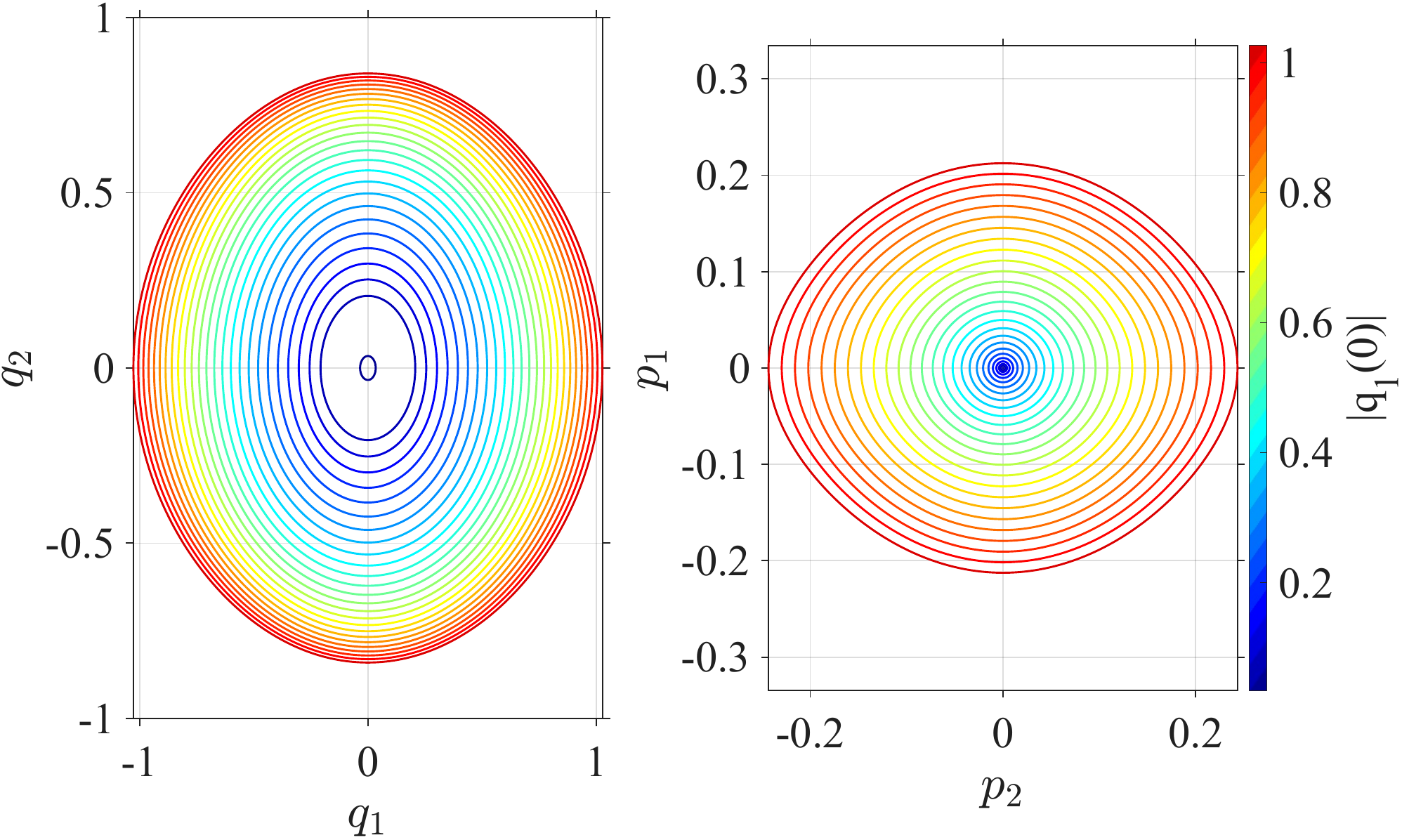}
\caption{Numerically computed periodic-orbit family around the origin.}
\label{fig:po_family_origin_pend}
%\end{figure}
%
%\begin{figure}[tb]
\centering
\includegraphics[width=0.6\linewidth]{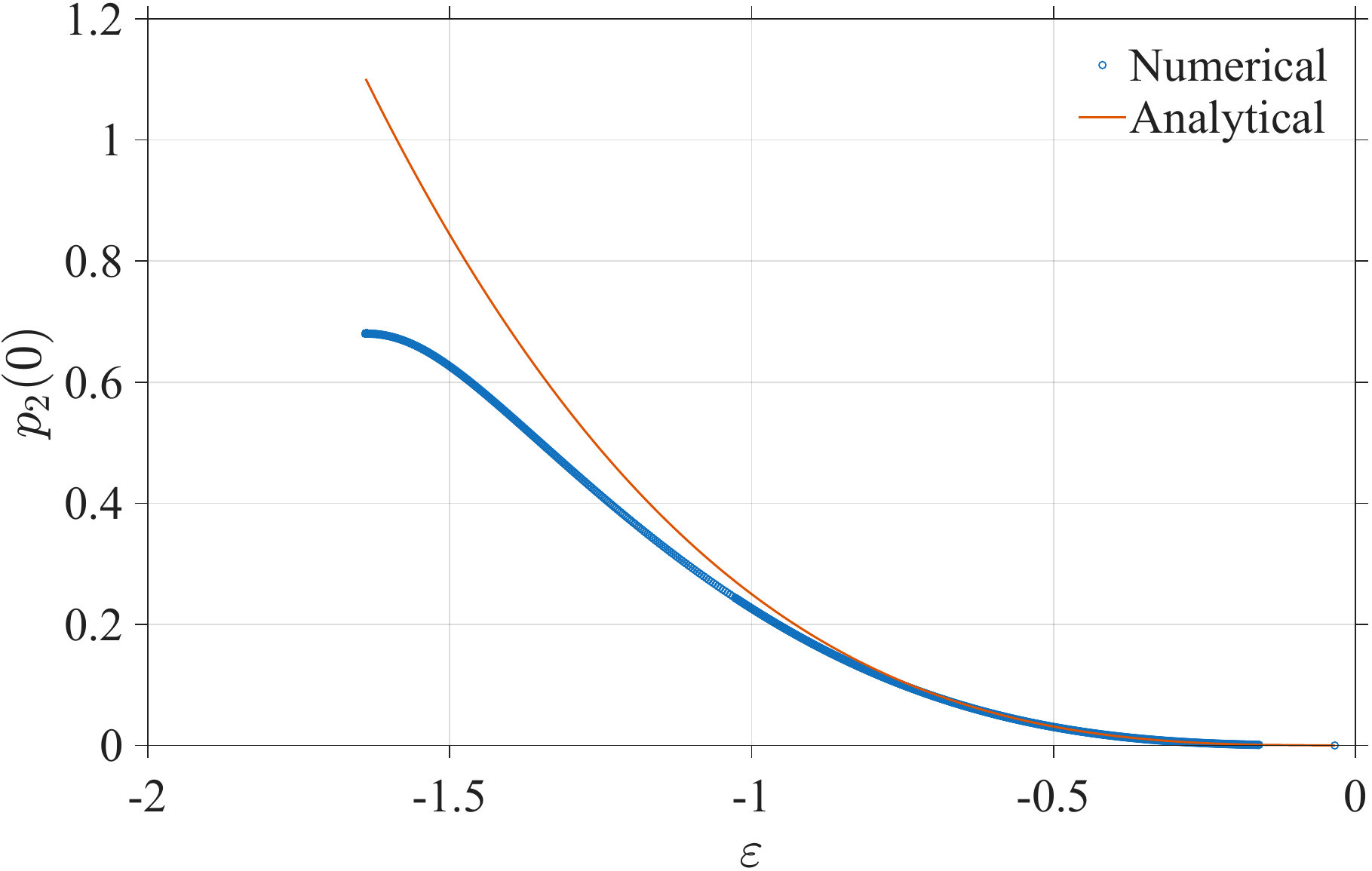}
\caption{Comparison of the initial costate $p_2(0)$ between the analytical approximation and the numerical periodic-orbit family.}
\label{fig:p2_compare_pend}
%\end{figure}
%
%\begin{figure}[tb]
\centering
\includegraphics[width=0.6\linewidth]{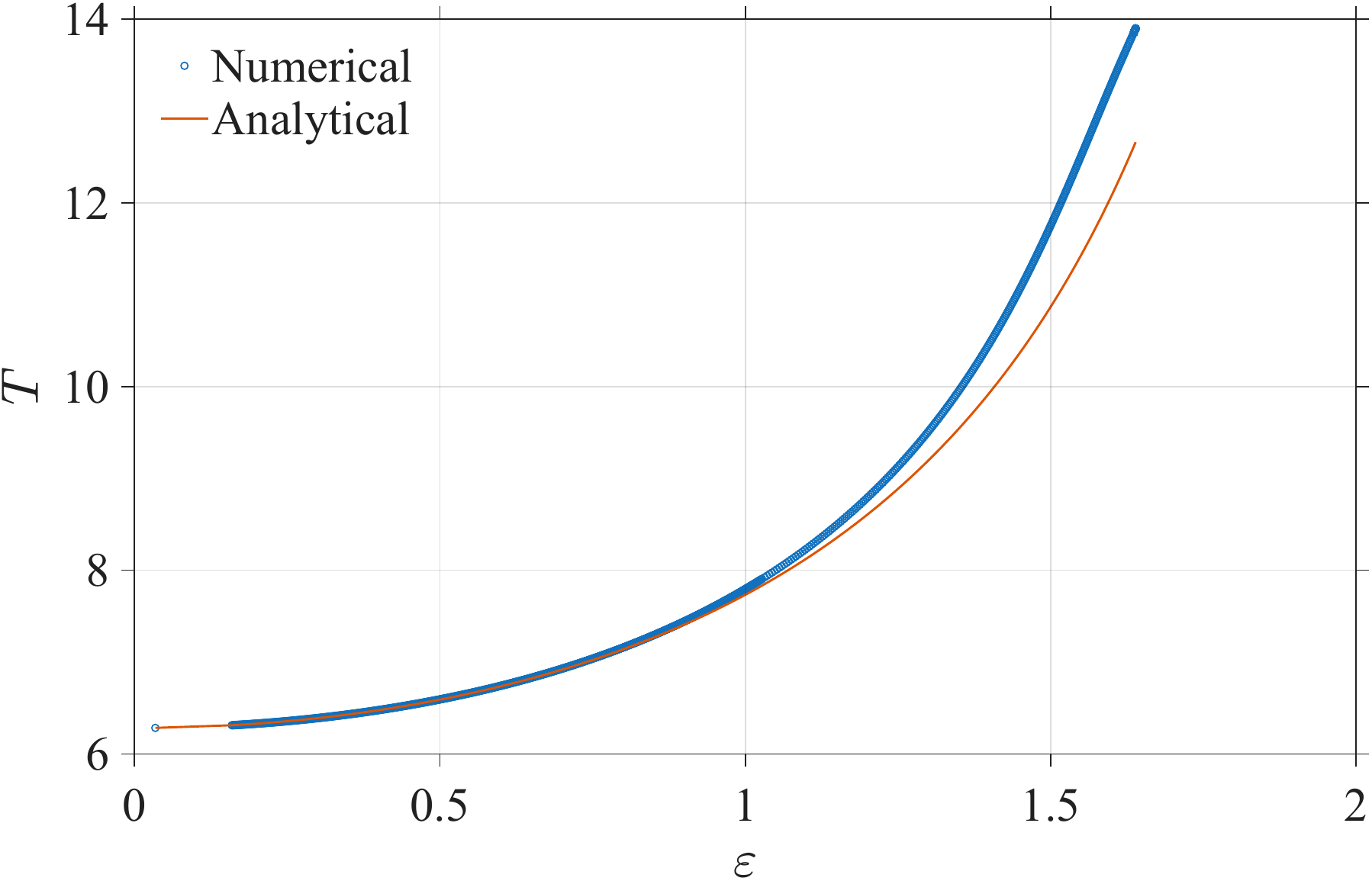}
\caption{Comparison of the orbital period between the analytical approximation and the numerical periodic-orbit family.}
\label{fig:period_compare_pend}
\end{figure}
\section{Application II: the planar $L_2$ equilibrium of Hill's restricted three-body problem}
\label{sec:hill3bp}

This section applies the general theory to the planar $L_2$ equilibrium of the planar Hill restricted three-body problem (Hill3BP).

\subsection{Optimal control problem}
\label{sec:hill3bp-setup}

The planar Hill3BP is described by the pseudo-potential
\cite{szebehely1967,scheeres2012}
\begin{equation}
U(x,y)=\frac1r+\frac32x^2,
\qquad
r=\sqrt{x^2+y^2},
\end{equation}
with the uncontrolled equations of motion
\begin{equation}
\ddot x-2\dot y=U_x,\qquad
\ddot y+2\dot x=U_y.
\end{equation}
Introducing
\begin{equation}
\boldsymbol r=
\begin{pmatrix}
x\\y
\end{pmatrix},
\qquad
\boldsymbol v=
\begin{pmatrix}
\dot x\\\dot y
\end{pmatrix},
\qquad
\boldsymbol x=
\begin{pmatrix}
\boldsymbol r\\
\boldsymbol v
\end{pmatrix},
\end{equation}
the controlled equations can be written as
\begin{equation}
\dot{\boldsymbol x}
=
\begin{pmatrix}
\boldsymbol v\\
2\boldsymbol J_a\boldsymbol v+\nabla U(\boldsymbol r)
\end{pmatrix}
+
\boldsymbol B\boldsymbol u,
\qquad
\boldsymbol B=
\begin{pmatrix}
\boldsymbol 0\\
\boldsymbol I
\end{pmatrix},
\label{eq:hill_controlled}
\end{equation}
where
\begin{equation}
\boldsymbol J_a=
\begin{pmatrix}
0&1\\
-1&0
\end{pmatrix}.
\end{equation}
For the minimum-energy cost
\begin{equation}
J
=
\frac12\int_{t_0}^{t_f}
\boldsymbol u^{\mathrm T}\boldsymbol u\,dt,
\end{equation}
Pontryagin's minimum principle gives the Hamiltonian
\begin{equation}
H(\boldsymbol x,\boldsymbol p,\boldsymbol u)
=
\boldsymbol p_r^{\mathrm T}\boldsymbol v
+
\boldsymbol p_v^{\mathrm T}
\left(
2\boldsymbol J_a\boldsymbol v
+\nabla U(\boldsymbol r)
+\boldsymbol u
\right)
+
\frac12\boldsymbol u^{\mathrm T}\boldsymbol u.
\end{equation}
The stationarity condition gives
\begin{equation}
\boldsymbol u^*=-\boldsymbol p_v.
\end{equation}
Substituting the optimal control yields the reduced Hamiltonian
\begin{equation}
H(\boldsymbol x,\boldsymbol p)
=
\boldsymbol p_r^{\mathrm T}\boldsymbol v
+
\boldsymbol p_v^{\mathrm T}
\left(
2\boldsymbol J_a\boldsymbol v
+\nabla U(\boldsymbol r)
\right)
-
\frac12\boldsymbol p_v^{\mathrm T}\boldsymbol p_v.
\end{equation}
Minimum-energy optimal control problems for the Hill3BP have been
considered in Refs.~\cite{tsuruta2024new,tsuruta2026periodic,pan2026fixed}.
The corresponding canonical equations are 
\begin{align}
\dot{\boldsymbol x}
&=
\left(
\frac{\partial H}{\partial\boldsymbol p}
\right)^{\mathrm T}
=
\begin{pmatrix}
\boldsymbol v\\
2\boldsymbol J_a\boldsymbol v
+
\nabla U(\boldsymbol r)
-
\boldsymbol p_v
\end{pmatrix},\label{eq:hill_canonical_x}
\\
\dot{\boldsymbol p}
&=
-
\left(
\frac{\partial H}{\partial\boldsymbol x}
\right)^{\mathrm T}
=
\begin{pmatrix}
-U_{rr}(\boldsymbol r)\boldsymbol p_v\\
-\boldsymbol p_r
+
2\boldsymbol J_a\boldsymbol p_v
\end{pmatrix}.
\label{eq:hill_canonical_p}
\end{align}
Linearizing Eqs.~\eqref{eq:hill_canonical_x} and \eqref{eq:hill_canonical_p} about the equilibrium
$(\boldsymbol x^*,\boldsymbol p^*)
=
((\boldsymbol r_{L_2},\boldsymbol0),(\boldsymbol0,\boldsymbol0))$
yields
\begin{equation}
\begin{pmatrix}
\delta\dot{\boldsymbol x}\\
\delta\dot{\boldsymbol p}
\end{pmatrix}
=
\boldsymbol M
\begin{pmatrix}
\delta\boldsymbol x\\
\delta\boldsymbol p
\end{pmatrix},
\end{equation}
with
\begin{equation}
\boldsymbol M=
\begin{pmatrix}
\bm0 & \bm I & \bm0 & \bm0\\
U_{rr}^* & 2\boldsymbol J_a & \bm0 & -\bm I\\
\bm0 & \bm0 & \bm0 & -U_{rr}^*\\
\bm0 & \bm0 & -\bm I & 2\boldsymbol J_a
\end{pmatrix},
\label{eq:M_hill}
\end{equation}
where $\boldsymbol r_{L_2}=\left(\left(\frac13\right)^{1/3},0\right)^{\mathrm T}$, and $U_{rr}^*$ denotes the Hessian of the Hill potential evaluated at $\boldsymbol r_{L_2}$.

The linearization of the natural (uncontrolled) in-plane dynamics at
$L_2$ has characteristic polynomial
\begin{equation}
\lambda^4-2\lambda^2-27=0,
\end{equation}
whose eigenvalues are $\lambda=\pm\lambda_s$ and
$\lambda=\pm i\omega_0$, where
$\lambda_s=\sqrt{1+2\sqrt7}$ and
$\omega_0=\sqrt{2\sqrt7-1}$.
These eigenvalues give the classical saddle$\times$center structure
\cite{gomez2001}.

The corresponding augmented $8\times8$ system has characteristic
polynomial
\begin{equation}
\bigl(\lambda^4-2\lambda^2-27\bigr)^2=0,
\end{equation}
confirming the exact spectral doubling.
For the center eigenvalues, direct evaluation of the pairing condition
introduced in Section~\ref{sec:theory} gives
\begin{equation}
\boldsymbol\ell^{\mathrm T}
\boldsymbol B\boldsymbol B^{\mathrm T}
\boldsymbol p_0
\neq0.
\end{equation}
Hence, by the Jordan-block criterion, each doubled center eigenvalue
forms a $2\times2$ Jordan block.

\subsection{Reduction to a second-order system}

The augmented first-order Hamiltonian system can be rewritten
exactly as the coupled second-order equations
\begin{equation}
\ddot{\boldsymbol r}
-2\boldsymbol J_a\dot{\boldsymbol r}
-\nabla U(\boldsymbol r)
=
-\boldsymbol p_v,
\qquad
\ddot{\boldsymbol p}_v
-2\boldsymbol J_a\dot{\boldsymbol p}_v
-\nabla^2U(\boldsymbol r)\boldsymbol p_v
=
\boldsymbol0,
\label{eq:reduced_vec}
\end{equation}
which form the starting point for the nonlinear analysis. 
Setting $\boldsymbol p_v\equiv\boldsymbol0$ reduces Eq.~\eqref{eq:reduced_vec} to the uncontrolled Hill equations.
The periodic family bifurcating from the $L_2$ equilibrium is the classical planar Lyapunov family, which constitutes the zero-adjoint branch.

\subsection{Nonlinear analysis}
\label{sec:hill3bp-nonlinear}

We next construct small-amplitude periodic solutions bifurcating from
the natural planar Lyapunov family. The derivation follows the
perturbation procedure established for the pendulum, and only the
solvability conditions required to determine the bifurcating branch
are summarized below.

\subsubsection{Lindstedt--Poincar\'e expansion}

Let
\begin{equation}
\boldsymbol r=\boldsymbol r_{L_2}+\boldsymbol X,
\qquad
\boldsymbol p_v=\boldsymbol P,
\end{equation}
where $\boldsymbol X$ denotes the displacement from the $L_2$
equilibrium. Introducing the strained time
\begin{equation}
\tau=\omega t,
\end{equation}
we expand
\begin{align}
\boldsymbol X(\tau)
&=
\varepsilon\boldsymbol X_1
+\varepsilon^2\boldsymbol X_2
+\varepsilon^3\boldsymbol X_3
+\cdots,
\\
\boldsymbol P(\tau)
&=
\varepsilon\boldsymbol P_1
+\varepsilon^2\boldsymbol P_2
+\varepsilon^3\boldsymbol P_3
+\cdots,
\\
\omega
&=
\omega_0
+\varepsilon\omega_1
+\varepsilon^2\omega_2
+\cdots,
\end{align}
where $\varepsilon$ denotes the oscillation amplitude.
The linear analysis in Section~\ref{sec:hill3bp-setup} shows that the
doubled center eigenvalues form Jordan blocks. As shown below, this
structure precludes a nonzero periodic adjoint component at linear
order, while the order at which a nonzero adjoint contribution first
appears is determined by the higher-order solvability conditions.

%---------------------------------------------------%
\paragraph{First order ($O(\varepsilon)$)}
%---------------------------------------------------%

Collecting the $O(\varepsilon)$ terms of Eq.~\eqref{eq:reduced_vec}
gives
\begin{align}
\omega_0^2\boldsymbol X_1''
-2\omega_0\boldsymbol J_a\boldsymbol X_1'
-U_{rr}^*\boldsymbol X_1
&=
-\boldsymbol P_1,
\label{eq:hill_first_state}
\\
\omega_0^2\boldsymbol P_1''
-2\omega_0\boldsymbol J_a\boldsymbol P_1'
-U_{rr}^*\boldsymbol P_1
&=
\boldsymbol0,
\label{eq:hill_first_adjoint}
\end{align}
where
\begin{equation}
U_{rr}^*
=
\begin{pmatrix}
9&0\\
0&-3
\end{pmatrix}.
\end{equation}
Equation~\eqref{eq:hill_first_adjoint} admits a periodic solution at
the natural frequency $\omega_0$. However, a nonzero component of
$\boldsymbol P_1$ in this mode resonantly forces
Eq.~\eqref{eq:hill_first_state}. As established by the Jordan-block
criterion in Section~\ref{sec:hill3bp-setup}, the resulting state
response contains a secular term and is therefore incompatible with
periodicity. Hence, $\boldsymbol P_1=\boldsymbol0$.

The first-order state equation then reduces to
\begin{equation}
\omega_0^2\boldsymbol X_1''
-2\omega_0\boldsymbol J_a\boldsymbol X_1'
-U_{rr}^*\boldsymbol X_1
=
\boldsymbol0.
\end{equation}
Choosing the phase and normalizing the leading-order $x$-amplitude,
we write
\begin{equation}
\boldsymbol X_1(\tau)
=
\begin{pmatrix}
\cos\tau\\
c_y\sin\tau
\end{pmatrix}.
\label{eq:X1_hill}
\end{equation}
Substitution gives
\begin{align}
\omega_0^2+9+2\omega_0c_y&=0,
\\
(\omega_0^2-3)c_y+2\omega_0&=0.
\end{align}
Eliminating $c_y$ yields
\begin{equation}
\omega_0^4+2\omega_0^2-27=0,
\end{equation}
and hence
\begin{equation}
\omega_0=\sqrt{2\sqrt7-1},
\qquad
c_y
=
-\frac{2\omega_0}{\omega_0^2-3}
=
-\frac{\omega_0^2+9}{2\omega_0}.
\label{eq:omega0_cy}
\end{equation}

%---------------------------------------------------%
\paragraph{Second order ($O(\varepsilon^2)$)}
%---------------------------------------------------%

The nonlinear terms are obtained by expanding the gradient of the Hill
potential about the equilibrium $\boldsymbol r_{L_2}$:
\begin{equation}
\begin{aligned}
\nabla U(\boldsymbol r)
={}&
U_{rr}^*\boldsymbol X
+\frac12U_{rrr}^*(\boldsymbol X,\boldsymbol X)
+\frac16U_{rrrr}^*
(\boldsymbol X,\boldsymbol X,\boldsymbol X)
+O(\|\boldsymbol X\|^4),
\end{aligned}
\label{eq:taylorU}
\end{equation}
where all derivative tensors are evaluated at
$\boldsymbol r_{L_2}$.
Using $\boldsymbol P_1=\boldsymbol0$ and collecting the
$O(\varepsilon^2)$ terms gives
\begin{align}
\omega_0^2\boldsymbol X_2''
-2\omega_0\boldsymbol J_a\boldsymbol X_2'
-U_{rr}^*\boldsymbol X_2
={}&
\boldsymbol F_2
-\boldsymbol P_2
-2\omega_0\omega_1\boldsymbol X_1''
+2\omega_1\boldsymbol J_a\boldsymbol X_1',
\label{eq:O2_state}
\\
\omega_0^2\boldsymbol P_2''
-2\omega_0\boldsymbol J_a\boldsymbol P_2'
-U_{rr}^*\boldsymbol P_2
={}&
\boldsymbol0,
\label{eq:O2_adjoint}
\end{align}
where
\begin{equation}
\boldsymbol F_2
=
\frac12U_{rrr}^*
(\boldsymbol X_1,\boldsymbol X_1).
\label{eq:F2}
\end{equation}
The periodic solution of Eq.~\eqref{eq:O2_adjoint} can be written as
\begin{equation}
\boldsymbol P_2
=
B_2\boldsymbol X_1,
\label{eq:P2_hill}
\end{equation}
where $B_2$ is an undetermined amplitude.
Since the quadratic forcing $\boldsymbol F_2$ contains only the
zeroth and second harmonics, the fundamental-harmonic part of the
right-hand side of Eq.~\eqref{eq:O2_state} arises only from
$\boldsymbol P_2$ and the frequency correction $\omega_1$ and has
coefficients
\begin{align}
r_{2,x}
&=
-B_2+2\omega_0\omega_1+2c_y\omega_1,
\\
s_{2,y}
&=
-B_2c_y+2\omega_0c_y\omega_1+2\omega_1.
\end{align}
For a bounded periodic solution, these coefficients must satisfy
\begin{equation}
r_{2,x}+c_y s_{2,y}=0,
\label{eq:O2_solvability}
\end{equation}
because the coefficient matrix associated with the fundamental
harmonic is singular, with $(1,c_y)^\top$ as its null vector.
Substitution gives
\begin{equation}
B_2
=
\left(
2\omega_0+\frac{4c_y}{1+c_y^2}
\right)\omega_1.
\label{eq:B2_omega1}
\end{equation}
Thus, the second-order analysis relates the adjoint amplitude $B_2$
to the first frequency correction $\omega_1$, but does not determine
them individually. Their values are determined by the third-order
analysis below.

%---------------------------------------------------%
\paragraph{Third order ($O(\varepsilon^3)$)}
%---------------------------------------------------%

At $O(\varepsilon^3)$, the adjoint equation gives
\begin{align}
\omega_0^2\boldsymbol P_3''
-2\omega_0\boldsymbol J_a\boldsymbol P_3'
-U_{rr}^*\boldsymbol P_3
=-2\omega_0\omega_1\boldsymbol P_2''
+2\omega_1\boldsymbol J_a\boldsymbol P_2'+
U_{rrr}^*(\boldsymbol X_1,\boldsymbol P_2).
\label{eq:O3_adjoint_general}
\end{align}
The last term contains only the zeroth and second harmonics and hence
does not contribute to the fundamental-harmonic consistency
condition. Using $\boldsymbol P_2=B_2\boldsymbol X_1$, the
fundamental-harmonic condition gives
\begin{equation}
B_2\omega_1
\left[
2\omega_0(1+c_y^2)+4c_y
\right]
=0.
\label{eq:O3_adjoint_sol}
\end{equation}
Since the quantity in brackets is nonzero, combining
Eq.~\eqref{eq:O3_adjoint_sol} with Eq.~\eqref{eq:B2_omega1} yields
\begin{equation}
\omega_1=0,
\qquad
B_2=0,
\qquad
\boldsymbol P_2=\boldsymbol0.
\label{eq:P2_omega1_zero}
\end{equation}
Thus, the first nonzero periodic adjoint contribution can occur at
$O(\varepsilon^3)$. The third-order adjoint equation then reduces to
\begin{equation}
\omega_0^2\boldsymbol P_3''
-2\omega_0\boldsymbol J_a\boldsymbol P_3'
-U_{rr}^*\boldsymbol P_3
=
\boldsymbol0,
\end{equation}
and its periodic solution can be written as
\begin{equation}
\boldsymbol P_3(\tau)
=
B_3\boldsymbol X_1(\tau),
\label{eq:P3}
\end{equation}
where $B_3$ is an undetermined amplitude.
With $\omega_1=0$ and $\boldsymbol P_2=\boldsymbol0$, the second-order
state correction contains only the zeroth and second harmonics and can
be written as
\begin{equation}
\boldsymbol X_2(\tau)
=
\begin{pmatrix}
p_{2,0}+p_{2,2}\cos2\tau\\
q_{2,2}\sin2\tau
\end{pmatrix}.
\label{eq:X2}
\end{equation}

At the same order, the state equation becomes
\begin{equation}
\omega_0^2\boldsymbol X_3''
-2\omega_0\boldsymbol J_a\boldsymbol X_3'
-U_{rr}^*\boldsymbol X_3
=
\boldsymbol F_3(\tau)-\boldsymbol P_3(\tau),
\label{eq:X3eq}
\end{equation}
where
\begin{align}
\boldsymbol F_3
=U_{rrr}^*(\boldsymbol X_1,\boldsymbol X_2)
+\frac16U_{rrrr}^*
(\boldsymbol X_1,\boldsymbol X_1,\boldsymbol X_1)
-2\omega_0\omega_2\boldsymbol X_1''
+2\omega_2\boldsymbol J_a\boldsymbol X_1'.
\label{eq:F3}
\end{align}
The forcing contains the fundamental and third harmonics. The third
harmonic is nonresonant, whereas the fundamental harmonic must satisfy
the same consistency condition as at second order. This gives
\begin{equation}
-(1+c_y^2)B_3
+
\left[
2\omega_0(1+c_y^2)+4c_y
\right]\omega_2
+
\left(
r_3^{\rm nl}+c_y s_3^{\rm nl}
\right)
=0,
\label{eq:SC1_unnormalized}
\end{equation}
where $r_3^{\rm nl}$ and $s_3^{\rm nl}$ denote the
fundamental-harmonic coefficients generated by the nonlinear terms in
Eq.~\eqref{eq:F3}. Equivalently,
\begin{equation}
B_3+\beta_1\omega_2+\gamma_1=0,
\label{eq:SC1}
\end{equation}
where
\begin{equation}
\beta_1
=
-2\omega_0-\frac{4c_y}{1+c_y^2},
\qquad
\gamma_1
=
-\frac{r_3^{\rm nl}+c_y s_3^{\rm nl}}{1+c_y^2}.
\label{eq:beta1gamma1}
\end{equation}

%---------------------------------------------------%
\paragraph{Fourth order ($O(\varepsilon^4)$)}
%---------------------------------------------------%

At $O(\varepsilon^4)$, the adjoint equation takes the form
\begin{equation}
\omega_0^2\boldsymbol P_4''
-2\omega_0\boldsymbol J_a\boldsymbol P_4'
-U_{rr}^*\boldsymbol P_4
=
\boldsymbol G_4(\tau),
\label{eq:P4eq}
\end{equation}
where
\begin{equation}
\boldsymbol G_4
=
\bigl[U_{rrr}^*\!\cdot\boldsymbol X_1\bigr]
\boldsymbol P_3.
\label{eq:G4}
\end{equation}
Since both $\boldsymbol X_1$ and $\boldsymbol P_3$ contain only the
fundamental harmonic, their product contains only the zeroth and
second harmonics. Thus, $\boldsymbol G_4$ has no resonant fundamental
component, and no additional solvability condition arises at this
order.
Accordingly, the periodic fourth-order adjoint correction can be
written as
\begin{equation}
\boldsymbol P_4(\tau)
=
B_3
\begin{pmatrix}
p_{4,0}+p_{4,2}\cos2\tau\\
q_{4,2}\sin2\tau
\end{pmatrix},
\label{eq:P4}
\end{equation}
where the coefficients are determined by matching the zeroth- and
second-harmonic terms.

%---------------------------------------------------%
\paragraph{Fifth order ($O(\varepsilon^5)$)}
%---------------------------------------------------%

At $O(\varepsilon^5)$, the adjoint equation becomes
\begin{equation}
\omega_0^2\boldsymbol P_5''
-2\omega_0\boldsymbol J_a\boldsymbol P_5'
-U_{rr}^*\boldsymbol P_5
=
\boldsymbol G_5(\tau),
\label{eq:P5eq}
\end{equation}
where
\begin{align}
\boldsymbol G_5
={}&
-2\omega_0\omega_2\boldsymbol P_3''
+2\omega_2\boldsymbol J_a\boldsymbol P_3'
+\bigl[U_{rrr}^*\!\cdot\boldsymbol X_1\bigr]\boldsymbol P_4
\nonumber\\
&+
\bigl[U_{rrr}^*\!\cdot\boldsymbol X_2\bigr]\boldsymbol P_3
+\frac12
\bigl[
U_{rrrr}^*(\boldsymbol X_1,\boldsymbol X_1)
\bigr]\boldsymbol P_3.
\label{eq:G5}
\end{align}
Because $\boldsymbol P_3$ contains the fundamental harmonic, whereas
$\boldsymbol X_2$ and $\boldsymbol P_4$ contain only the zeroth and
second harmonics, $\boldsymbol G_5$ contains only the fundamental and
third harmonics. The third harmonic is nonresonant, while the
fundamental harmonic must satisfy the same consistency condition as
at second order.
Since every term in $\boldsymbol G_5$ is proportional to $B_3$, this
condition factorizes as
\begin{equation}
B_3
\left(
\alpha_2+\beta_2\omega_2
\right)
=
0,
\label{eq:SC2}
\end{equation}
where $\alpha_2$ and $\beta_2$ are determined by the
fundamental-harmonic coefficients of Eq.~\eqref{eq:G5}.

Equations~\eqref{eq:SC1} and \eqref{eq:SC2} therefore define two
branches. The solution
\begin{equation}
B_3=0
\end{equation}
corresponds to the continuation of the zero-adjoint planar Lyapunov
family. For the nontrivial branch, $B_3\neq0$, Eq.~\eqref{eq:SC2}
gives
\begin{equation}
\omega_2
=
-\frac{\alpha_2}{\beta_2},
\label{eq:omega2_solution}
\end{equation}
and substitution into Eq.~\eqref{eq:SC1} determines
\begin{equation}
B_3
=
-\beta_1\omega_2-\gamma_1.
\label{eq:B3_solution}
\end{equation}
Accordingly, the adjoint variable has the expansion
\begin{equation}
\boldsymbol P(\tau)
=
\varepsilon^3\boldsymbol P_3(\tau)
+\varepsilon^4\boldsymbol P_4(\tau)
+O(\varepsilon^5),
\end{equation}
and the nontrivial branch departs from the zero-adjoint branch with
the leading cubic scaling
\begin{equation}
\boldsymbol P(\tau)
=
\varepsilon^3 B_3\boldsymbol X_1(\tau)
+O(\varepsilon^4).
\label{eq:P_leading}
\end{equation}

\subsubsection{Numerical evaluation of the coefficients}
\label{sec:numerical_evaluation}

Substituting the derivatives of the Hill potential at $L_2$ into the
preceding expressions gives the coefficients summarized in
Table~\ref{tab:hill_coefficients}.
Using these values, the two solvability conditions
Eqs.~\eqref{eq:SC1} and \eqref{eq:SC2} become
\begin{equation}
B_3
=
10.1516+3.0067\,\omega_2,
\label{eq:Bomega}
\end{equation}
and
\begin{equation}
B_3
\left(
343.880+33.951\,\omega_2
\right)
=
0.
\label{eq:Bomega_fifth}
\end{equation}
For the nontrivial branch, $B_3\neq0$, and hence
\begin{equation}
\omega_2\approx-10.129,
\qquad
B_3\approx-20.303.
\label{eq:hill_omega2_B3}
\end{equation}
Thus, the leading-order frequency and adjoint expansions of the
nontrivial branch are
\begin{align}
\omega
&=
\omega_0-10.129\,\varepsilon^2+O(\varepsilon^3),
\\
\boldsymbol P(\tau)
&=
-20.303\,\varepsilon^3\boldsymbol X_1(\tau)
+O(\varepsilon^4).
\label{eq:hill_final_asymptotic}
\end{align}

\begin{table}[tbph]
\centering
\caption{Numerical coefficients obtained from the
Lindstedt--Poincar\'e analysis.}
\label{tab:hill_coefficients}
\begin{tabular}{ccc}
\hline
Order & Coefficient & Value \\
\hline
$O(\varepsilon^2)$
& $p_{2,0}$ & $-2.98960$ \\
& $p_{2,2}$ & $1.29968$ \\
& $q_{2,2}$ & $0.70951$ \\
\hline
$O(\varepsilon^3)$
& $\beta_1$ & $-3.0067$ \\
& $\gamma_1$ & $-10.1516$ \\
\hline
$O(\varepsilon^4)$
& $p_{4,0}$ & $-5.97921$ \\
& $p_{4,2}$ & $2.59936$ \\
& $q_{4,2}$ & $1.41902$ \\
\hline
$O(\varepsilon^5)$
& $\alpha_2$ & $343.880$ \\
& $\beta_2$ & $33.951$ \\
\hline
\end{tabular}
\end{table}

\subsubsection{Comparison with numerical periodic orbits}

The analytical approximation is validated by comparison with
numerically computed periodic orbits obtained by differential
correction.
Figure~\ref{fig:po_family_L2_hill3bp} shows the family of controlled periodic
orbits around the planar $L_2$ equilibrium.
Figure~\ref{fig:pvx_compare_hill3bp} compares the analytical and numerical
values of the initial adjoint variable as a function of the initial
amplitude $\varepsilon$. Figure~\ref{fig:period_compare_hill3bp} compares the
corresponding periods.
The analytical approximation accurately reproduces the numerical periodic-orbit family for small amplitudes. As the amplitude increases, neglected higher-order nonlinear terms become significant, leading to increasing deviations from the numerical solutions.

\begin{figure}[tb]
\centering
\includegraphics[width=0.8\linewidth]{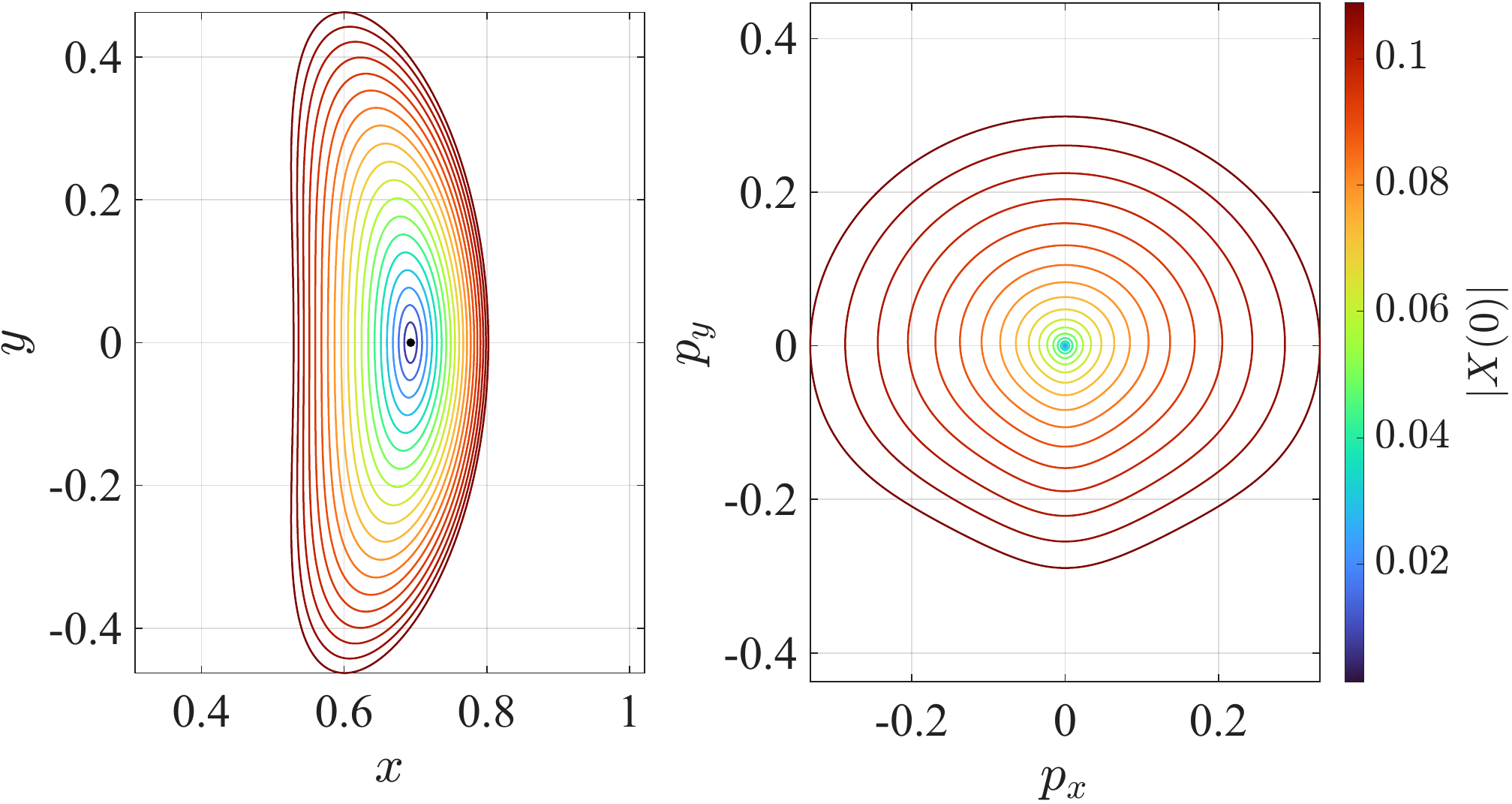}
\caption{Numerically computed periodic-orbit family around $L_2$.}
\label{fig:po_family_L2_hill3bp}
\centering
\includegraphics[width=0.6\linewidth]{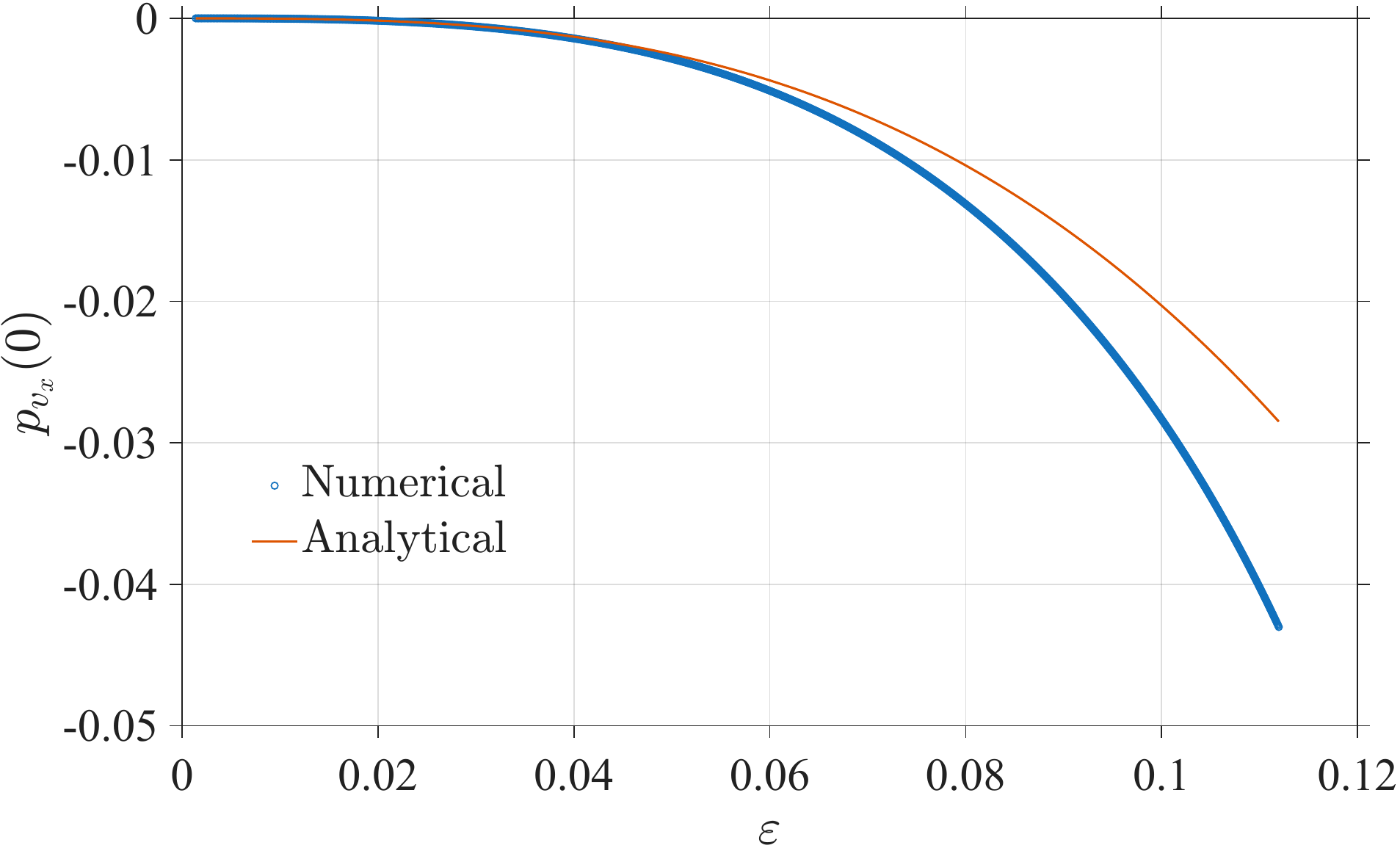}
\caption{Comparison between the analytical approximation and the
numerically computed periodic-orbit family.}
\label{fig:pvx_compare_hill3bp}
\centering
\includegraphics[width=0.6\linewidth]{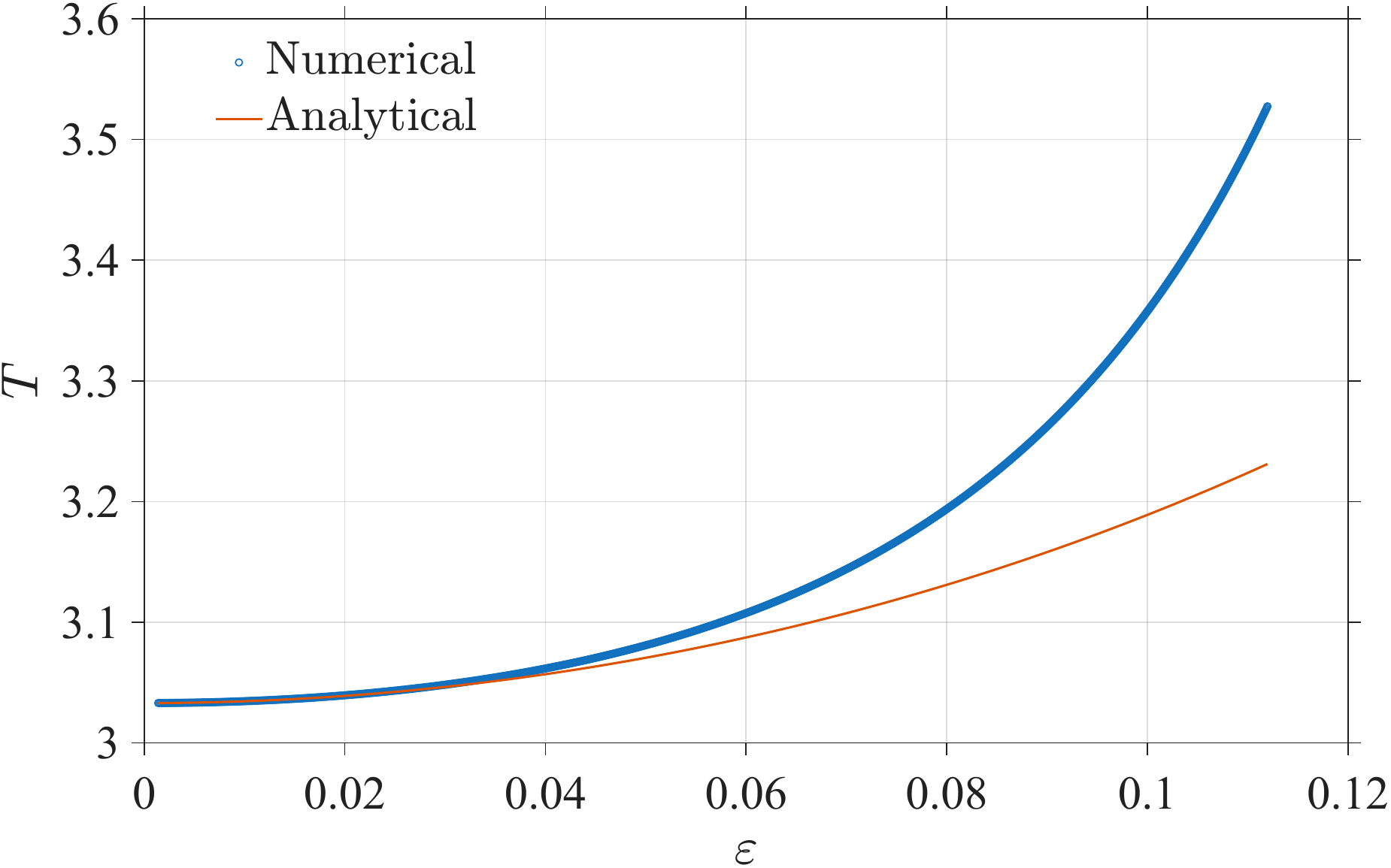}
\caption{Comparison of the analytical period with the numerically computed periodic-orbit family.}
\label{fig:period_compare_hill3bp}
\end{figure}
\section{Discussion}
\label{sec:discussion}
\subsection{Physical interpretation}
\label{sec:discussion-physical}

In both examples, the leading-order costate oscillates at the same natural frequency as the uncontrolled system, a direct consequence of the Jordan block established in Theorem~2. 
Therefore, it resonates with the natural motion, producing a secular response rather than a periodic solution. 
Since no nonlinear correction appears at $O(\varepsilon)$ or $O(\varepsilon^2)$, there is no mechanism capable of removing this resonant forcing, and the linear analysis therefore yields only the trivial branch ($\boldsymbol p = \boldsymbol0$).

The first opportunity to remove the resonance arises at $O(\varepsilon^3)$, although through different nonlinear mechanisms in the two examples. 
For the pendulum, the restoring force satisfies $\sin q_1 = q_1-\frac16q_1^3+O(q_1^5)$, so the first nonlinear correction appears only at cubic order. 
For Hill3BP, the gravitational force expands as $\nabla U(\boldsymbol r)=U_{rr}^*\delta\boldsymbol r+ \frac12 U_{rrr}^*(\delta\boldsymbol r,\delta\boldsymbol r) + \frac16 U_{rrrr}^*(\delta\boldsymbol r,\delta\boldsymbol r,\delta\boldsymbol r) +\cdots$,  with $\delta\boldsymbol r=O(\varepsilon)$. 
The quadratic term produces only a mean shift and a second harmonic, and therefore does not contribute to the resonant fundamental mode. 
The first resonant contribution arises at $O(\varepsilon^3)$ from the interaction between the quadratic correction $\boldsymbol X_2$ and the fundamental harmonic, together with the cubic gravitational term.
At this order, the nonlinear terms, the costate amplitude, and the frequency correction $\omega=\omega_0+\varepsilon^2\omega_2+\cdots$ enter simultaneously.
Their balance removes the resonant forcing, yielding the first solvability condition ($B=2\omega_2+\tfrac18$ for the pendulum and Eq.~\eqref{eq:SC1} for Hill3BP). 
The first solvability condition involving the cubic-order costate amplitude arises at this order. Together with the higher-order solvability condition, it selects a nontrivial branch whose costate scales cubically with the oscillation amplitude.
Consequently, the first optimal-control-induced periodic orbit appears only at cubic order.

\subsection{Relation to previous studies}

The Jordan-block degeneracy identified here is closely related to that
reported by Sakamoto~\cite{sakamoto2024}, who analyzed the stability
of the natural periodic-orbit family in the zero-costate subspace.
In contrast, the present work investigates the bifurcation of periodic
solutions with nonzero costates and shows that a nontrivial controlled
branch emerges only through higher-order nonlinear effects, with
cubic scaling in the oscillation amplitude.
\section{Conclusion}
\label{sec:conclusion}

This paper has investigated minimum-energy optimal control systems
obtained by augmenting Hamiltonian dynamics through Pontryagin's
minimum principle.
A general theory was established showing that the optimal-control
augmentation doubles the natural spectrum and, under a simple pairing
condition, produces a $2\times2$ Jordan block at each simple center
eigenvalue.
The theory was applied to a pendulum and Hill's restricted three-body
problem.
Although the two systems have different nonlinear structures, both
exhibit the same mechanism: the Jordan-block degeneracy precludes
periodic solutions with nonzero costates at linear order, whereas
higher-order nonlinear effects allow a nontrivial controlled periodic
branch to bifurcate from the natural periodic-orbit family.
In both examples, the costate, and hence the optimal control, exhibits
cubic scaling with the oscillation amplitude.
These results identify a general mechanism by which minimum-energy
optimal control modifies the local bifurcation structure of periodic
orbits near Hamiltonian equilibria.
\bibliographystyle{elsarticle-num}
\bibliography{references}

@book{pontryagin1962,
  author    = {Pontryagin, L. S. and Boltyanskii, V. G. and Gamkrelidze, R. V. and Mishchenko, E. F.},
  title     = {The Mathematical Theory of Optimal Processes},
  publisher = {Interscience Publishers},
  address   = {New York},
  year      = {1962}
}

@book{bryson1975,
  author    = {Bryson, Arthur E., Jr. and Ho, Yu-Chi},
  title     = {Applied Optimal Control: Optimization, Estimation, and Control},
  publisher = {Hemisphere Publishing Corporation},
  address   = {Washington, DC},
  year      = {1975}
}

@book{liberzon2012,
  author    = {Liberzon, Daniel},
  title     = {Calculus of Variations and Optimal Control Theory: A Concise Introduction},
  publisher = {Princeton University Press},
  address   = {Princeton},
  year      = {2012},
  doi       = {10.1515/9781400842643}
}

@book{agrachev2004,
  author    = {Agrachev, Andrei A. and Sachkov, Yuri L.},
  title     = {Control Theory from the Geometric Viewpoint},
  series    = {Encyclopaedia of Mathematical Sciences},
  volume    = {87},
  publisher = {Springer},
  address   = {Berlin, Heidelberg},
  year      = {2004},
  doi       = {10.1007/978-3-662-06404-7}
}

@article{weinstein1973,
  author  = {Weinstein, Alan},
  title   = {Normal Modes for Nonlinear Hamiltonian Systems},
  journal = {Inventiones Mathematicae},
  volume  = {20},
  number  = {1},
  pages   = {47--57},
  year    = {1973},
  doi     = {10.1007/BF01405263}
}

@article{moser1976,
  author  = {Moser, J{\"u}rgen},
  title   = {Periodic Orbits Near an Equilibrium and a Theorem by Alan Weinstein},
  journal = {Communications on Pure and Applied Mathematics},
  volume  = {29},
  number  = {6},
  pages   = {727--747},
  year    = {1976},
  doi     = {10.1002/cpa.3160290613}
}

@book{meyer2009,
  author    = {Meyer, Kenneth R. and Hall, Glen R. and Offin, Dan},
  title     = {Introduction to Hamiltonian Dynamical Systems and the N-Body Problem},
  edition   = {2},
  series    = {Applied Mathematical Sciences},
  volume    = {90},
  publisher = {Springer},
  address   = {New York},
  year      = {2009},
  doi       = {10.1007/978-0-387-09724-4}
}

@article{richardson1980,
  author  = {Richardson, David L.},
  title   = {Analytic Construction of Periodic Orbits about the Collinear Points},
  journal = {Celestial Mechanics},
  volume  = {22},
  number  = {3},
  pages   = {241--253},
  year    = {1980},
  doi     = {10.1007/BF01229511}
}

@book{lyapunov1992,
  author     = {Lyapunov, A. M.},
  title      = {The General Problem of the Stability of Motion},
  translator = {Fuller, A. T.},
  publisher  = {Taylor \& Francis},
  address    = {London},
  year       = {1992},
  isbn       = {9780748400621},
  note       = {Originally published in Russian in 1892}
}

@article{sakamoto2024,
  author        = {Sakamoto, Noboru},
  title         = {Maximally Degenerate {Floquet} Structure and Possible
                   Nonexistence of Optimal Control in a Pendulum
                   Swing-Up Problem},
  journal       = {arXiv preprint arXiv:2403.16156},
  year          = {2026}
}

@article{tsuruta2024new,
  author  = {Tsuruta, Ayano and Bando, Mai and Hokamoto, Shinji and Scheeres, Daniel J.},
  title   = {New Equilibria and Dynamic Structures Under Continuous Optimal Feedback Control},
  journal = {Journal of Guidance, Control, and Dynamics},
  volume  = {47},
  number  = {10},
  pages   = {2029--2040},
  year    = {2024},
  doi     = {10.2514/1.G008270}
}

@article{tsuruta2026periodic,
  author  = {Tsuruta, Ayano and Pan, Shanshan and Bando, Mai and
             Hokamoto, Shinji and Scheeres, Daniel J.},
  title   = {Periodic and Quasi-Periodic Orbits Induced by Optimal Control
             in the Hill Three-Body Problem},
  journal = {Celestial Mechanics and Dynamical Astronomy},
  volume  = {138},
  pages   = {38},
  year    = {2026},
  doi     = {10.1007/s10569-026-10313-2}
}

@article{pan2026fixed,
  author  = {Pan, Shanshan and Bando, Mai and Tsuruta, Ayano and
             Scheeres, Daniel J.},
  title   = {Fixed-Time Optimal Periodic Orbit Control in the Hill Three-Body Problem},
  journal = {Journal of Guidance, Control, and Dynamics},
  year    = {2026},
  doi     = {10.2514/1.G009916}
}

@book{gomez2001,
  author    = {G{\'o}mez, Gerard and Llibre, Jaume and
               Mart{\'i}nez, Regina and Sim{\'o}, Carles},
  title     = {Dynamics and Mission Design Near Libration Points,
               Volume I: Fundamentals: The Case of Collinear Libration Points},
  series    = {World Scientific Monograph Series in Mathematics},
  volume    = {2},
  publisher = {World Scientific},
  year      = {2001},
  doi       = {10.1142/4402}
}

@article{trelat2018steady,
  author  = {Tr{\'e}lat, Emmanuel and Zhang, Can and Zuazua, Enrique},
  title   = {Steady-State and Periodic Exponential Turnpike Property for
             Optimal Control Problems in Hilbert Spaces},
  journal = {SIAM Journal on Control and Optimization},
  volume  = {56},
  number  = {2},
  pages   = {1222--1252},
  year    = {2018},
  doi     = {10.1137/16M1097638}
}

@book{scheeres2012,
  author    = {Scheeres, Daniel J.},
  title     = {Orbital Motion in Strongly Perturbed Environments: Applications to Asteroid, Comet and Planetary Satellite Orbiters},
  publisher = {Springer},
  address   = {Berlin, Heidelberg},
  year      = {2012},
  doi       = {10.1007/978-3-642-03256-9}
}

@book{szebehely1967,
  author    = {Szebehely, Victor},
  title     = {Theory of Orbits: The Restricted Problem of Three Bodies},
  publisher = {Academic Press},
  year      = {1967},
  doi       = {10.1016/B978-0-12-395732-0.X5001-6}
}

@book{Conway2010,
  editor    = {Conway, Bruce A.},
  title     = {Spacecraft Trajectory Optimization},
  series    = {Cambridge Aerospace Series},
  volume    = {29},
  publisher = {Cambridge University Press},
  year      = {2010},
  doi       = {10.1017/CBO9780511778025},
  isbn      = {9780521518505}
}

@article{Kelly2017,
  author  = {Kelly, Matthew P.},
  title   = {An Introduction to Trajectory Optimization:
             How to Do Your Own Direct Collocation},
  journal = {SIAM Review},
  volume  = {59},
  number  = {4},
  pages   = {849--904},
  year    = {2017},
  doi     = {10.1137/16M1062569}
}

@article{Sakamoto2008,
  author  = {Sakamoto, Noboru and van der Schaft, Arjan J.},
  title   = {Analytical Approximation Methods for the Stabilizing Solution of the Hamilton--Jacobi Equation},
  journal = {IEEE Transactions on Automatic Control},
  volume  = {53},
  number  = {10},
  pages   = {2335--2350},
  year    = {2008},
  doi     = {10.1109/TAC.2008.2006113}
}

@book{vanderMeer1985,
  author    = {van der Meer, J. C.},
  title     = {The {H}amiltonian {H}opf Bifurcation},
  series    = {Lecture Notes in Mathematics},
  volume    = {1160},
  publisher = {Springer-Verlag},
  address   = {Berlin},
  year      = {1985},
  doi       = {10.1007/BFb0080357}
}

@article{vanGils1990,
  author  = {van Gils, S. A. and Krupa, M. and Langford, W. F.},
  title   = {Hopf Bifurcation with Non-Semisimple 1:1 Resonance},
  journal = {Nonlinearity},
  volume  = {3},
  number  = {3},
  pages   = {825--850},
  year    = {1990},
  doi     = {10.1088/0951-7715/3/3/013}
}

\end{document}